%% file: ENA8-LNCS.tex
\documentclass[runningheads]{llncs}
\usepackage{graphicx}
\usepackage{amsfonts}
\usepackage{amssymb,epic, graphicx}
\usepackage[cmex10]{amsmath}
\usepackage{array}
\usepackage{mdwmath}
\usepackage{mdwtab}
\usepackage{eqparbox}
\usepackage[tight,footnotesize]{subfigure}
\usepackage{algorithm}
\usepackage{algorithmic}
\usepackage{makecell}
\usepackage{multirow}
\usepackage{threeparttable}
\usepackage{multicol}
\usepackage[toc,page]{appendix}
\usepackage{appendix}

\newtheorem{cor}[theorem]{Corollary}
\graphicspath{{./}{figs/}}

\makeatletter
\newenvironment{breakablealgorithm}
{
	\begin{center}
		\refstepcounter{algorithm}
		\hrule height.8pt depth0pt \kern2pt
		\renewcommand{\caption}[2][\relax]{
			{\raggedright\textbf{\ALG@name~\thealgorithm} ##2\par}%
			\ifx\relax##1\relax 
			\addcontentsline{loa}{algorithm}{\protect\numberline{\thealgorithm}##2}%
			\else 
			\addcontentsline{loa}{algorithm}{\protect\numberline{\thealgorithm}##1}%
			\fi
			\kern2pt\hrule\kern2pt
		}
	}{
		\kern2pt\hrule\relax
	\end{center}
}
\makeatother
\begin{document}
\title{The Elliptic Net Algorithm Revisited
}
%
%
\author{Shiping Cai\inst{1}\and
Zhi Hu\inst{2} \and
Zheng-An Yao\inst{1} \and Chang-An Zhao\inst{1,3}}
\institute{School of Mathematics, Sun Yat-sen University, Guangzhou 510275, P.R. China \and
 School of Mathematics and Statistics, Central South University, P.R. China \and Guangdong Key Laboratory of Information Security, Guangzhou 510006, P.R. China}
\maketitle              

\begin{abstract}
\input{tex/abstract}
\end{abstract}

\input{tex/main}
%
%
%
\bibliographystyle{splncs04}
\bibliography{mybibliography,pairing}
\input{tex/appendix}

\end{document}

%% file: tex/abstract.tex
Pairings have been widely used since their introduction to cryptography. They can be applied to identity-based encryption, tripartite Diffie-Hellman key agreement, blockchain and other cryptographic schemes. The Acceleration of  pairing computations is crucial for these cryptographic schemes or protocols. In this paper, we will focus on the Elliptic Net algorithm which can compute pairings in polynomial time, but it requires more storage than Miller's algorithm. We use several methods to speed up the Elliptic Net algorithm. Firstly, we eliminate the inverse operation in the improved Elliptic Net algorithm. In some circumstance, this finding can achieve further improvements. Secondly, we apply lazy reduction technique to the Elliptic Net algorithm, which helps us achieve a faster implementation. Finally, we propose a new derivation of the formulas for the computation of the Optimal Ate pairing on the twisted curve. Results show that the Elliptic Net algorithm can be significantly accelerated especially on the twisted curve. The algorithm can be $80\%$ faster than the previous ones on the twisted 381-bit BLS12 curve and $71.5\%$ faster on the twisted 676-bit KSS18 curve respectively. 

%% file: tex/main.tex
\section{Introduction}\label{sec1}
Pairings as mathematical primitives can offer efficient solutions to some special difficult problems in cryptography \cite{MarkJoyebook}. Nowadays, pairings still play a vital role in some areas. In the blockchain, pairings can be applied for the zero-knowledge succinct non-interactive argument of knowledge (zk-SNARK) \cite{Marlin,Groth2016}. Moreover, pairings can be used for the compression of public keys in the isogeny-based cryptosystem \cite{DualIsogenies2019}. 

The implementation of pairings is a key operation in these applications. The Weil pairing, the Tate pairing and their variants such as the Ate pairing \cite{HSV,Matsuda2007Optimised}, the R-ate pairing \cite{Lee2009Efficient}, and the Optimal Ate pairing \cite{Vercauteren2009Optimal} are used in some cryptographic schemes. It is well known that pairings can be computed by Miller's algorithm \cite{miller1986short,Miller2004} which was proposed in 1986. Miller's algorithm has been optimized a lot since 2000, and it has been developed to be in a relatively mature stage. There also exists another polynomial time algorithm to compute pairings, i.e., the Elliptic Net algorithm. This algorithm was proposed in 2007 by Stange~\cite{Stange2007} who first defined elliptic nets and gave a relationship between elliptic nets and the Tate pairing. We abbreviate this original Elliptic Net algorithm to ENA. Compared with Miller's algorithm, the Elliptic Net algorithm requires more computational costs while it can be implemented efficiently on a personal computer. Furthermore, there is no inverse operation involved in affine coordinates in the Elliptic Net algorithm. Therefore, the implementation of this algorithm is simple and intuitive.

Elliptic nets are generated by elliptic divisibility sequences which were first studied by Morgan Ward \cite{ward1948memoir} in 1948. These sequences arise from any choice of an elliptic curve and rational points on that curve. For more information about elliptic divisibility sequences see \cite{einsiedler_everest_ward_2001}.  The method called \textbf{Double-and-Add} for updating each value of an elliptic divisibility sequence in polynomial time which was proposed by Rachel Shipsey \cite{Shipsey2000Elliptic}. 

Pairings can be computed using elliptic nets of rank $2$. The ENA was used to compute the Tate pairing originally~\cite{Stange2007}. Then the explicit formulas for computing some variants of the Tate pairing using the ENA were given \cite{Tangchunming2014,CPEN}. In 2015, an improved version of the ENA was proposed \cite{Chen2015AnIO}. We abbreviate this algorithm to IENA in this work. The IENA can perform well if the parameter of the Miller loop has low Hamming weight. Fortunately, the most popular paring-friendly curves all meet this condition. Due to the particularity of the structure of elliptic nets, an parallel strategy for the ENA to compute pairings was proposed \cite{Onuki2016Faster}.

Elliptic nets of rank $1$ can be applied to scalar multiplication, and it is an algorithm that can resist side-channel attacks. Kanayama \emph{et al.} \cite{KANAYAMA2014} adopted the ENA to compute scalar multiplication using elliptic nets of rank $1$, i.e., division polynomials. Besides, there are some other works about scalar multiplication based on elliptic nets. An optimized version of scalar multiplication algorithm using division polynomials was proposed in~\cite{2017A}, which saved $four$ multiplications at each iteration by using the equivalence of elliptic nets. Based on these previous works, Rao \emph{et al.} \cite{8890309} proposed a modified algorithm based on elliptic nets to compute scalar multiplication. 

However, the efficiency of the Elliptic Net algorithm still needs to be avoided. In order to shorten the gap between the Elliptic Net algorithm and Miller's algorithm in efficiency, we develop several methods. Firstly, we analyze the properties of elliptic nets and conclude that the inverse operation in the IENA can be eliminated. Secondly, we construct the Optimal Ate pairing on the twisted curve and discuss the relationship between the Optimal Ate pairing on the original elliptic curve and that on the twisted curve with divisor and pull-back. This is a new derivation of the formulas for pairing computation which is entirely on the twisted curve. Thirdly, lazy reduction technique is employed in our implementation to get a further improvement. The specific contributions of this work are:
\begin{itemize}
	\item We explore how to eliminate the inverse operation in the IENA. For the IENA, an inverse operation is involved at addition step in the Double-and-Add algorithm. In this paper, we get a result in the updating process of the IENA. If all the values of an elliptic net in the current state are multiplied by a non-zero fixed value, then the value of the reduced Tate pairing or its variants can not be changed. This finding means that the inverse operation can be replaced by several multiplications. The implementation indicates that the IENA works well if it is further modified by this trick. Besides, this trick contributes to the scalar multiplication algorithm in \cite{8890309}.
	
	\item The idea of twists is employed to speed up the Elliptic Net algorithm. Twists of elliptic curves are deeply applied to Miller's algorithm, which can significantly decrease the amount of multiplications. More detailed descriptions see \cite{silver,HSV}. Throughout the process of Miller's algorithm, Costello \emph{et al.} \cite{2010Faster} explored the pairing computation which is entirely on the twisted curve. Based on these works, we use the Elliptic Net algorithm to compute the Optimal Ate pairing entirely on the twisted curve. We will provide a new proof based on the theory of divisors and pull-back, which has a strength that it not depends on the process of Miller's algorithm. Hence, this is totally different from the previous work. Furthermore, we give the explicit formulas of the line function of the Optimal Ate pairing on the twisted curve. We boost the performance of the Elliptic Net algorithm on a 381-bit BLS12 curve at 128-bit security level and a 676-bit KSS18 curve at 192-bit security level by using twists \cite{2018Updating}. Twists of elliptic curves allow us to transfer the operations from $\mathbb{F}_{q^k}$ to the proper subfield of $\mathbb{F}_{q^{k}}$, which significantly reduces the total amount of multiplications. 
	
	\item We adopt lazy reduction technique \cite{lim2000fast} which only performs one reduction for the sum of several multiplications to the Elliptic Net algorithm. Lazy reduction was first introduced in quadratic extension field arithmetic for Miller's algorithm by Michael Scott~\cite{scott2006implementing} and further developed  in~\cite{DKPC2011}. In the Elliptic Net algorithm, we observe that there are many terms have the form $A \cdot B - C \cdot D$, which inspires us to apply lazy reduction for this algorithm. In our implementation, lazy reduction reduces by around 27\% the number of modular reductions.
\end{itemize}

We conclude that pairings can be efficiently computed with the Elliptic Net algorithm. Even though it is still slower than Miller's algorithm, the ratio of the cost of the Elliptic Net algorithm to Miller's algorithm is reduced from more than $9$ to less than $2$ after the modification in this work.

The rest of this paper is organized as follows. Section~\ref{sec2} gives an overview of pairings, twists of elliptic curves and the Elliptic Net algorithm. In Section~\ref{sec3}, we replace an inverse operation by several multiplications in the IENA. Section~\ref{sec4} analyzes the Ate pairing and Optimal Ate pairing on the twisted curve that are computed by the Elliptic Net algorithm. In Section~\ref{sec5}, we apply lazy reduction technique to the Elliptic Net algorithm. The implementation and efficiency analysis are discussed in Section~\ref{sec6}. Section~\ref{sec7} concludes the paper.

\section{Preliminaries}\label{sec2}
In this section, we will give the definition of the Tate pairing and the (Optimal) Ate pairing. A brief description of twists of elliptic curves and the Elliptic Net algorithm will also be provided.
\subsection{Pairings}\label{sec2.1}
Let $\mathbb F_q$ be a finite field where $q=p^m\,(m\in \mathbb{Z}^+)$ and $p$ is a prime. Let $E$ be an elliptic curve defined over $\mathbb F_q$. We denote the $q$-power Frobenius endomorphism on $E$ by $\pi _q$. The number of points on $E/\mathbb{F}_q$ is given by $\#E(\mathbb F_q)=q+1-t$, where $t$ is the Frobenius trace of $\pi_q$. 

Choose a large prime $r$ with $r|\#E(\mathbb F_q)$, and let $k \in \mathbb Z^+$ be the embedding degree with respect to $r$ such that $r|q^k-1$ but $r^2\nmid q^k-1$. Choose $P \in E(\mathbb{F}_q)[r]$ and $Q \in \mathbb E(\mathbb{F}_{q^k})[r]$. The $r$-th roots of unity in $\mathbb F_{q^k}$, which we denote by $\mu_r$.

For an integer $i$ and a point $S$ on $E$, let $f_{i,S}$ be a rational function such that $${\rm Div}(f_{i,S})=i(S)-(iS)-(i-1)(\infty).$$ In particular, ${\rm Div}(f_{r,P})=rD_P=r(P)-r(\infty)$. Then the \textbf{reduced Tate pairing} \cite{GrangerHess2007} is defined as
$$\begin{aligned}
	Tate:\mathbb E(\mathbb{F}_q)[r] \times \mathbb E(\mathbb{F}_{q^k})[r] &\rightarrow \mu_r\\ (P,Q)&\mapsto Tate(P,Q)=f_{r,P}(Q)^{{q^k-1}/r}.
\end{aligned}$$

Furthermore, if we choose $P$ and $Q$ in specific subgroups of $E[r]$, the pairing computation can be sped up. Define
$$\mathbb G_1 \triangleq E[r]\bigcap Ker(\pi _q -[1]),\,\mathbb G_2 \triangleq E[r]\bigcap Ker(\pi _q -[q]).$$
Choose $P\in \mathbb{G}_1$ and $Q\in \mathbb{G}_2$. Let $T=t-1$. We can define a pairing when $r\nmid (T^k-1)/r$: 
$$\begin{aligned}
	Ate_E:\mathbb G_2 \times \mathbb G_1 &\rightarrow \mu_r\\ (Q,P)&\mapsto Ate_E(Q,P)=f_{T,Q}(P)^{{q^k-1}/r},
\end{aligned}$$ which is called the \textbf{Ate pairing} \cite{HSV}.

The Ate pairing is a variant of the Tate pairing, and the length of the Miller loop is short\cite{Matsuda2007Optimised,Lee2009Efficient}. 
The Optimal Ate pairing allows us to obtain the shortest loop length~\cite{Vercauteren2009Optimal}. We have the following theorem \cite{Vercauteren2009Optimal,ZZH08,ZZH07}.
\begin{theorem}\label{th2.2}
	Let $\lambda =\alpha r$ with $r \nmid \alpha$. We have $\lambda =\sum_{i=0}^{\varphi(k)}c_iq^i$, where $\varphi(k)$ is the Euler function of $k$, then we can define a bilinear map
	$$\begin{aligned}
		Opt_E:\mathbb G_2 \times \mathbb G_1 &\rightarrow \mu_r\\ (Q,P)&\!\mapsto Opt_E(Q,P)\!=(\!\prod _{i=0}^{\varphi(k)-1}\!f_{c_i,Q}^{q^i}(P)\!\cdot \!\prod_{i=0}^{\varphi(k)-1}\!\frac{l_{[s_{i+1}]Q,[c_iq^i]Q}(P)}{v_{[s_i]Q}(P)})^{{q^k-1}/r},
	\end{aligned}$$ 
	where $s_i=\sum_{j=i}^{\varphi(k)}c_jq^j$. If $\alpha k q^{k-1} \neq\left(\left(q^{k}-1\right) / r\right) \sum_{i=0}^{\varphi(k)-1} i c_{i} q^{i-1}(\bmod\,r),$
	then $Opt_E$ is non-degenerate. We call $Opt_E$ as the \textbf{Optimal Ate pairing}.
\end{theorem}                         
The explicit expression of the Optimal Ate pairing depends on the family type of pairing-friendly curves. In this work, we mainly consider the implementation of the Optimal Ate pairing on the BLS12 and KSS18 curves. More specific information will be discussed in Section~\ref{sec6}. 

\subsection{Twists of Elliptic Curves}
\begin{definition} \label{def2.2}
	A twist of degree \(d\) of \(E\) is an elliptic curve \(E^{\prime}\) defined over $\mathbb{F}_{q^{k/d}}$. We can define an isomorphism \(\Psi_{d}\) over \(\mathbb{F}_{q^{k}}\) from \(E^{\prime}\) to \(E\) with \(d\) is\textbf{ minimal}:
	\[
	\Psi_{d}: E^{\prime}\left(\mathbb{F}_{q^{k/d}}\right) \longrightarrow E\left(\mathbb{F}_{q^{k}}\right).
	\]
\end{definition}
The potential degree $d$ is $2,3,4$ or $6$ \cite{MarkJoyebook,silver}. For the BLS12 and KSS18 curves, $E'$ is a twist of degree $6$ of $E$. Let \(\xi \in \mathbb{F}_{q^{k/6}}\). For the M-type and D-type twists \cite{softimp} with degree $d=6$, the corresponding isomorphism $\Psi_6$ is given as follows:
\begin{equation}\label{eq13}
	\begin{aligned}
		M-type:\,&E':y^{2}=x^{3}+b\xi &\Psi_6: E'\rightarrow E:(x, y) \mapsto\left(\xi^{-1/3} x, \xi^{-1 / 2} y\right),\\
		D-type:\,&E':y^{2}=x^{3}+b / \xi &\Psi_6: E'\rightarrow E:(x, y) \mapsto\left(\xi^{1/3} x, \xi^{1 / 2} y\right).
	\end{aligned}
\end{equation}

Furthermore, we have the following theorem for the Tate pairing:
\begin{theorem}\cite{2005Advances}
	Let $E_1/\mathbb{F}_q$ be an elliptic curve. Choose $r_0|\#E_1(\mathbb{F}_q)$. Suppose that the embedding degree with respect to $q$ and $r_0$ is $k$. There exists an isogeny $\phi:E_1 \rightarrow E_2$ and $\hat{\phi}$ is the dual of $\phi$, where $E_2$ is an elliptic curve over $\mathbb{F}_{q^k}$. Choose $P \in E_1(\mathbb{F}_q)[n]$ and $Q \in E_2(\mathbb{F}_{q^k})$. We have $e(P,\phi (Q))=e(\hat{\phi}(P),Q)$. 
\end{theorem}

Notice that $\Psi_d$ is an isogeny of degree $1$. If we denote the dual of $\Psi_d$ by $\hat{\Psi}_d$, then $\hat{\Psi}_d\circ \Psi_d =[1]$, i.e., $\Psi_d^{-1}=\hat{\Psi}_d$. By Definition~\ref{def2.2}, choose $P \in E(\mathbb{F}_q)[r]$ and $Q' \in E'(\mathbb{F}_{q^{k/d}})$. We can compute pairings $(Opt)Ate_{E'}(\hat{\Psi}_d(P),Q')$ on the twisted curve $E'$. And the loop length is the same as $(Opt)Ate_E(P,\Psi_d(Q'))$ which is computed on the original curve $E$.

Furthermore, define
$$\Phi_d=\Psi_d^{-1}\circ \pi_q \circ \Psi_d.$$
One can verify that $\Phi_d$ is a group isomorphism from $E'$ to $E'$ over $\mathbb{F}_{q^k}$ \cite{GalEnd}, which can be used in Section~\ref{sec4}.
\subsection{The Elliptic Net algorithm}\label{sec2.3}
An elliptic net satisfies some certain recurrence relation which is a map $W$ from a finitely generated free Abelian group $A$ to an integral domain $R$. An elliptic net of rank 1 satisfies the following recurrence relation:
\begin{equation}\label{eq1}
	\begin{aligned}
		&W(\alpha+\beta+\delta,0)W(\alpha-\beta,0)W(\gamma+\delta,0)W(\gamma,0)\\
		&+W(\beta+\gamma+\delta,0)W(\beta-\gamma,0)W(\alpha+\delta,0)W(\alpha,0)\\
		&+W(\gamma+\alpha+\delta,0)W(\gamma-\alpha,0)W(\beta+\delta,0)W(\beta,0)=0,
	\end{aligned}
\end{equation}
where $\alpha,\,\,\beta,\,\gamma,\,\delta \in A$.
Let $E_0:y^2=x^3+Ax+B$ be a short Weierstrass curve over $\mathbb{F}_q$, where $4A^3+27B^2 \neq 0$. And the characteristic of $\mathbb{F}_q$ is not equal to $2$ or $3$.
\subsubsection{Scalar Multiplication}
For each $n \in \mathbb{Z}^{+}$, we can define division polynomials $\psi_n\in \mathbb{Z}[A,B,x,y]$ as follows \cite{silver}.
\begin{equation*}
	\begin{aligned}
		\psi_0&=0,\psi_1=1,\psi_2=2y,\\
		\psi_3&=3x^4+6Ax^2+12Bx-A^2,\\
		\psi_4&=4y(x^6+5Ax^4+20Bx^3-5A^2x^2-4ABx-8B^2-A^3),\\
		\psi_{2n+1}\psi_1&=\psi_{n+2}\psi_n^3-\psi_{n-1}\psi_{n+1}^3\,\,(n\geq 2),\\
		\psi_{2n}\psi_2&=\psi_n(\psi_{n+2}\psi_{n-1}^2-\psi_{n-2}\psi_{n+1}^2)\,\,(n \geq 3).
	\end{aligned}
\end{equation*}

Division polynomials are elliptic nets of rank $1$, i.e., $W(i,0)=\psi_i,\,\forall i \in \mathbb{Z}$. They can be used to compute scalar multiplication.

 Choose $P=(x_P,y_P)\in E_0(\mathbb{F}_q)$, and define two polynomials $\zeta_n,\,\omega_n$. These formulas can be used to compute $[n]P=(x_{nP},y_{nP})$ as follows.
\begin{equation}\label{eq2.2}
	[n]P=\left(\frac{\zeta_n(P)}{\psi_n(P)^2},\frac{\omega_n(P)}{\psi_n(P)^3}\right),
\end{equation}
where
\begin{equation*}
	\begin{aligned}
		\zeta_n&=x_P\psi_n^2-\psi_{n+1}\psi_{n-1},\\
		4y_P\omega_n&=\psi_{n+2}\psi_{n-1}^2-\psi_{n-2}\psi_{n+1}^2.
	\end{aligned}
\end{equation*}
Equation~(\ref{eq2.2}) can be represented by elliptic nets of rank $1$ \cite{KANAYAMA2014}:
\begin{equation}\label{eq2.3}
	\begin{aligned}
		x_{nP}&=x_P-\frac{W(n-1,0)W(n+1,0)}{W(n,0)^2},\\
		y_{nP}&=\frac{W(n-1,0)^2W(n+2,0)-W(n+1,0)^2W(n-2,0)}{4y_PW(n,0)^3}.
	\end{aligned}
\end{equation}
\subsubsection{Pairing Computation}
Elliptic nets of rank $2$ is applied to pairing computation. The relationship between the Tate pairing and an elliptic net is given below.

\begin{theorem}\cite{Stange2007} \label{th1}
	Choose $P\in E(\mathbb{F}_q)[r]$ and $Q \in E(\mathbb{F}_{q^k})[r]$ such that $[r]P=\infty$. If $W_{P,Q}$ is the elliptic net associated to $E$, $P$, $Q$, then we have
	
	\begin{equation*}
		f_{r,P}(D_Q)=\frac{W_{P,Q}(r+1,1)W_{P,Q}(1,0)}{W_{P,Q}(r+1,0)W_{P,Q}(1,1)}.
	\end{equation*}
\end{theorem}

According to Equation~(\ref{eq1}), we obtain the explicit formulas to update the values of an elliptic net. We can compute the Tate pairing in polynomial time if the initial values of an elliptic net are given. For simplicity, we abbreviate $W_{P,Q}(n,s)$ to $W(n,s)$. In \cite{Stange2007}, they defined a block that consists of a first vector of eight consecutive terms centered on term $W(i,0)$ and a second vector of three consecutive terms centered on $W(i,1)$, where $i \in \mathbb{Z}$. 

Assume that $W(1,0)=W(0,1)=1$. For the first vector, all of $W(n,0)$ terms can be updated by two formulas as follows.
\begin{equation}\label{eq2}
	W(2i-1,0)=W(i+1,0)W(i-1,0)^3-W(i-2,0)W(i,0)^3,
\end{equation}
\begin{equation}\label{eq3}
	\begin{split}
		W(2i,0)=(W(i,0)W(i+2,0)W(i-1,0)^2\\-W(i,0)W(i-2,0)W(i+1,0)^2)/W(2,0).
	\end{split}
\end{equation}
For the second vector, we need the following formulas to update the $W(n,1)$ terms.
\begin{equation}\label{eq4}
	\begin{split}
		W(2i-1,1)=(W(i+1,1)W(i-1,1)W(i-1,0)^2\\-W(i,0)W(i-2,0)W(i,1)^2)/W(1,1),
	\end{split}
\end{equation}
\begin{equation}\label{eq5}
	\begin{split}
		W(2i,1)=(W(i-1,1)W(i+1,1)W(i,0)^2\\-W(i-1,0)W(i+1,0)W(i,1)^2),
	\end{split}
\end{equation}
\begin{equation}\label{eq6}
	\begin{split}
		W(2i+1,1)=(W(i-1,1)W(i+1,1)W(i+1,0)^2\\-W(i,0)W(i+2,0)W(i,1)^2)/W(-1,1),
	\end{split}
\end{equation}
\begin{equation}\label{eq7}
	\begin{split}
		W(2i+2,1)=(W(i+1,0)W(i+3,0)W(i,1)^2\\-W(i-1,1)W(i+1,1)W(i+2,0)^2)/W(2,-1).
	\end{split}
\end{equation}

For some certain conditions, $W(2,0)$ can be changed to $1$ by  the equivalence of elliptic nets \cite{2017A}.

\floatname{algorithm}{Algorithm}
\renewcommand{\algorithmicrequire}{\textbf{INPUT:}}
\renewcommand{\algorithmicensure}{\textbf{OUTPUT:}}
\begin{algorithm}
	\caption{The improved Elliptic Net algorithm \cite{Chen2015AnIO}}\label{alg1}
	\begin{algorithmic}[1] 
		\REQUIRE Initial terms $a=W(2,0),\,b=W(3,0),\,c=W(4,0),\,d=W(2,1),\,e=W(-1,1),\,f=W(2,-1),\,g=W(1,1),\,h=W(2,1)$ of the Elliptic Net algorithm satisfies $W(1,0)=W(0,1)=1$ and $n=(d_ld_{l-1}...d_0)_2 \in \mathbb Z$ with $d_l=1$ and $d_i \in \{0,1\}$ for $0 \leq i \leq l-2$
		\ENSURE $W(n,0),W(n,1)$
		\STATE $V \gets [[-a,-1,0,1,a,b,c],[1,g,d]]$
		\FOR{$i = l-1 \to 0$}
		\IF{$d_i=0$}
		\STATE $V \gets Double(V)$
		\ELSE
		\STATE $V \gets DoubleAdd(V)$
		\ENDIF
		\ENDFOR
		\RETURN{$V[0,3],V[1,1]$}
	\end{algorithmic}
\end{algorithm}
The IENA is shown in Algorithm~\ref{alg1}. Generally, updating a block centered on $i$ to a block centered on $2i$ is called the Double step, and updating a block centered on $i$ to a block centred on $2i+1$ is called the DoubleAdd step, which is represented by $Double(V)$ and $DoubleAdd(V)$ respectively. The algorithm to compute the process of line 2-8 in Algorithm~\ref{alg1} is called the Double-and-Add algorithm. If we just need to compute scalar multiplication, then we only need to update the first vector by Equation~(\ref{eq2})-(\ref{eq3}). We do not use the IENA to compute scalar multiplication here. There exists an inversion if we need to compute the DoubleAdd step in the IENA. But for the scalar multiplication, the scalar $n$ is random and we can not ensure that $n$ is an integer with low Hamming weight. Notice that $4$ multiplications can be saved at each iteration in the Double-and-Add algorithm if $gcd(p-1,3)=1$ \cite{2017A}. In this work, we will improve the Double-and-Add algorithm for scalar multiplication in two situations in \cite{8890309}.

Note that twists of elliptic curves have been applied for accelerating Miller's algorithm successfully. The situation of operations entirely on the twisted curve $E_0'$ was proposed \cite{2010Faster}. Their derivation of the Ate pairing entirely on the twisted curve heavily relies on the process of Miller's algorithm. Pairings entirely on the twisted curve can also be computed by the Elliptic Net algorithm, but it still needs to be verified. 

\section{Elimination of the Inverse Operation}\label{sec3}
In the IENA, an inverse operation is always involved at addition step. We will show how to eliminate this inverse operation in this section, i.e., replace the inversion by few multiplications.

When a block centered on $i$ is updated to a block centered on $2i+1$, the value $W(2i+4,0)$ satisfies the following recursive formula:
\begin{equation}\label{eq11}
	W(2 i+4,0)=\frac{W(2 i+3,0) W(2 i+1,0) W(2,0)^{2}-W(3,0) W(1,0) W(2 i+2,0)^{2}}{W(2 i, 0)}.
\end{equation}

From Equation~(\ref{eq11}), we need to compute the inverse element of $W(2i,0)$. To eliminate this inverse operation, we multiply $W(\lambda,0)_{2i-3\leq \lambda \leq 2i+4}$ by $W(2i,0)$ simultaneously when the bit is equal to $1$. We have the following theorem to support this approach.

\begin{theorem}\label{th3}
	Let $W(\lambda,0)_{i-3\leq \lambda \leq i+3}$, $W(\lambda,1)_{i-1\leq \lambda \leq i+1} \in \mathbb{F}_{q^k}$ be the current state of an elliptic net. 
	\begin{enumerate}
		\item For $\alpha \in \mathbb{F}_{q^k}^*$, if $W(\lambda,0)_{i-3\leq \lambda \leq i+3}$ are multiplied by $\alpha$, i.e., 	
		$$\hat W(\lambda,0)_{i-3\leq \lambda \leq i+3}=\alpha\cdot W(\lambda,0)_{i-3\leq \lambda \leq i+3},$$ then in the next state
		\begin{equation*}
			\begin{aligned}
				&\hat W(\lambda,0)_{2i-3\leq \lambda \leq 2i+3}=\alpha^4 \cdot W(\lambda,0)_{2i-3\leq \lambda \leq 2i+3},\\
				&\hat W(\lambda,1)_{2i-1\leq \lambda \leq 2i+1}=\alpha^2 \cdot W(\lambda,1)_{2i-1\leq \lambda \leq 2i+1}.
			\end{aligned}
		\end{equation*}
		Furthermore, if $\alpha \neq 0$ is chosen to be in a proper subfield of $\mathbb{F}_{q^k}$, then the value of the reduced Tate pairing or its variants can not be changed.
		\item For $\alpha \!\in \mathbb{F}_{q^k}^*$, if $W(\lambda,0)_{i-3\leq \lambda \leq i+3}$ and $W(\lambda,1)_{i-1 \leq \lambda \leq i+1}$ are multiplied by $\alpha$ , then in the next state all the terms of this elliptic net will be multiplied by $\alpha^4$, and the value of the reduced Tate pairing or its variants can not be changed.
	\end{enumerate}	
\end{theorem}

\begin{proof}
	Let us consider $\hat W(2i-1,0)$ first. 
	
	Note that the recursive formula for $W(2i-1,0)$ is  
	\begin{equation}
		\begin{aligned}
			W\!(2i\!-\!1, 0) &= W\!(i\! + \!1, 0)W\!(i\! - \!1, 0)^3\! -\! W\! (i\! -\! 2, 0)W \!(i, 0)^3.
		\end{aligned}
	\end{equation}
	We multiply $W(\lambda,0)_{i-3\leq \lambda\leq i+3}$ by $\alpha$, then the new updated $\hat{W}(2i-1,0)$ should be 
	\begin{equation}
		\begin{aligned}
			\hat{W}\!(2i\!-\!1, 0) &=\alpha^4( W\!(i\! + \!1, 0)W\!(i\! - \!1, 0)^3\! -\! W\! (i\! -\! 2, 0)W \!(i, 0)^3)\\
			&=\alpha^4\cdot W(2i-1,0). 
		\end{aligned}
	\end{equation}
	
	Similarly, we can show that the new updated 
	$\hat{W}(2i,0)=\alpha^4\cdot W(2i,0).$
	This finishes the proof for the first assertion.
	
	Then we consider the second vector. Note that there are only two values of the first vector involved for computing each $W(\lambda,1)_{2i-1\leq \lambda \leq  2i+2}$. The new updated $\hat{W}(\lambda,1)_{2i-1\leq \lambda \leq 2i+2}$ will be multiplied by $\alpha^2$.
	
	Therefore, the value of the new pairing is equal to the product of the original pairing value and a fixed power of $\alpha$. However, if the constant $\alpha$ is chosen to be in a proper subfield of $\mathbb{F}_{q^k}$, then the final exponentiation will eliminate the value of the fixed power of the constant $\alpha$. So the value of the reduced Tate pairing or its variants can not be changed even if  all the values of $W(\lambda,0)_{i-3\leq \lambda \leq i+3}$  in the state are multiplied by a non-zero fixed value $\alpha$. 
	
	Now we prove the second part of this theorem. 
	In Theorem~\ref{th1}, we know that
	\begin{equation*}
		f_{r,P}(D_Q)=\frac{W_{P,Q}(r+1,1)W_{P,Q}(1,0)}{W_{P,Q}(r+1,0)W_{P,Q}(1,1)}.
	\end{equation*}
	
	If we multiply $W(\lambda,0)_{i-3 \leq \lambda \leq i+3}$ and $W(\lambda,1)_{i-1 \leq \lambda \leq i+1}$ by $\alpha$, where $\alpha$ is any non-zero value, then we have:
	\begin{equation*}
		\begin{aligned}
			f_{r,P}(D_Q)&=\frac{\alpha^\ell W_{P,Q}(r+1,1)W_{P,Q}(1,0)}{\alpha^\ell W_{P,Q}(r+1,0)W_{P,Q}(1,1)}\\
			&=\frac{W_{P,Q}(r+1,1)W_{P,Q}(1,0)}{W_{P,Q}(r+1,0)W_{P,Q}(1,1)}\\
		\end{aligned}
	\end{equation*}
	for some integer $\ell$. This means that if we multiply all values in the updating block by a fixed non-zero value, the ratio of $\frac{W_{P,Q}(r+1,1)}{W_{P,Q}(r+1,0)}$ can not be changed.\qed
\end{proof}

\begin{remark}
	We just consider the situation at the Double step. At the the  DoubleAdd step, the conclusion can be verified similarly. 
\end{remark}
\begin{remark}
	Theorem~\ref{th3} can also be applied for any pairing-friendly curves while we may have to multiply both dimension of vectors. In this situation, we cost $8$ multiplications and the result can not be changed.
\end{remark}

Until now, we have shown how to replace the inverse operation by several multiplications. For some popular pairing-friendly curves, we have a friendly situation. Take the BLS12 curve we used in this work as an example, then there is a proposition which is helpful to our algorithm. The related parameters of the BLS12 curve can be seen in Section~\ref{sec6.1} and the towering scheme is shown as follows.
\begin{itemize}
	\item $\mathbb F_{q^2}=\mathbb F_q[u]/\langle u^2-\beta \rangle$, where $\beta = -1$;
	\item $\mathbb F_{q^6}=\mathbb F_{q^2}[v]/\langle v^3-\xi \rangle$, where $\xi = u+1$;
	\item $\mathbb F_{q^{12}}=\mathbb F_{q^6}[\omega]/\langle \omega^2-v \rangle $.
\end{itemize}

\begin{proposition}\label{pr2}
	Choose $P \!\in E_0(\mathbb{F}_q)$ and $Q'\!=(x_Q,y_Q) \!\in E_0'(\mathbb{F}_{q^2})$. 
	
	For $W_{\Psi_6(Q'),P}(s,0)(s\in \mathbb{Z})$ , if $s$ is odd, then $W_{\Psi_6(Q'),P}(s,0)$ is in the proper subfield of $\mathbb{F}_{q^{12}}$; If $s$ is even, then $W_{\Psi_6(Q'),P}(s,0)$ belongs to  $ \mathbb{F}_{q^{12}}$. Furthermore, let $W_{\Psi_6(Q'),P}(s,0)=a_0+a_1\omega,\,a_0,a_1\in \mathbb{F}_{q^6}$, $a_0=0$ if $s$ is even.
\end{proposition}
\begin{proof}
	We abbreviate $W_{\Psi_6(Q'),P}(s,0)$ to $W_{\Psi_6(Q')}(s,0)$.
	
	Note that $W_{\Psi_6(Q')}(s,0)=\psi_s\in \mathbb{Z}[x,y,A,B]$, where $\psi_s$ is a division polynomial.
	Therefore, we just verify the proposition in two situations according to Section~3.2 in \cite{2008Elliptic}:
	
	\begin{enumerate}
		\item  Assume that $s$ is odd, then $\psi_s$ is a polynomial in $\mathbb{Z}[x,y^2,A,B]$. For the short Weierstrass curve $y^2=x^3+Ax+B$, $y^2$ can be replaced by polynomials in $x$. Furthermore, $Q'\in E'$ and the $x$-coordinate of $\Psi_6(Q')$ is $x_Qv \in \mathbb{F}_{q^6}$. Therefore, $W_{\Psi_6(Q')}(s,0)$ is always in a proper subfield of $\mathbb{F}_{q^{12}}$.
		\item If $s$ is even, then $\psi_s$ is a polynomial in $2y\mathbb{Z}[x,y^2,A,B]$. And the $y$-coordinate of  $\Psi_6(Q')$ is $y_Qv\omega \in \mathbb{F}_{q^{12}}$, so $\psi_s\in 2y\mathbb{Z}[x,y^2,A,B]$ can be written as $a_1\omega$, where $a_1\in \mathbb{F}_{q^6}$. 
	\end{enumerate}\qed
\end{proof}

From Proposition~\ref{pr2}, if $W(\lambda,0)_{i-3 \leq \lambda \leq i+3}$ are multiplied by $\alpha=a_1\omega$, where $a_1\in \mathbb{F}_{q^6}$, then both $\alpha^2$ and $\alpha^4$ will always be in $\mathbb{F}_{q^6}$. From Theorem~\ref{th3}, in the next state the value of $W_{\Psi_6(Q'),P}(2s,0)$ or $W_{\Psi_6(Q'),P}(2s+1,0)$ will be multiplied by $\alpha^4$. And $W_{\Psi_6(Q'),P}(2s,1)$ or $W_{\Psi_6(Q'),P}(2s+1,1)$ will be multiplied by $\alpha^2$. Therefore, if the last iteration is the doubling step, then the value of the reduced Tate pairing or its variants can not be changed.

Moreover, the last iteration of the Miller loop on the BLS12 curve will always invoke the doubling step. Hence, we can avoid the inversion in the IENA. This means that we can use $5$ multiplications to eliminate $1$ inversion. Therefore, the effect of our method always works well when the cost of $1$ inverse operation is more than that of $5$ multiplications. 

For the situation of scalar multiplication algorithm based on elliptic nets, we have the following corollary. 
\begin{cor}\label{cor1}
	Choose $P=(x_P,y_P)\in E_0(\mathbb{F}_q)$. Let $W(\lambda,0)_{i-3\leq \lambda \leq i+4}\in \mathbb{F}_{q}$ be the current state of an elliptic net which is associate to $E_0,\,P$. For $\alpha \in \mathbb{F}_{q}^*$, if $W(\lambda,0)_{i-3\leq \lambda \leq i+4}$ are multiplied by $\alpha$, i.e.,  $\hat W(\lambda,0)_{i-3\leq \lambda \leq i+4}=\alpha\cdot W(\lambda,0)_{i-3\leq \lambda \leq i+4}$. Then the value of the scalar multiplication $[n]P=(x_{nP},y_{nP})(n \in \mathbb{Z}^+)$ can not be changed.
\end{cor}
\begin{proof}
	From Theorem~\ref{th3}, we know that in the next state each $W(\lambda,0)$ will be multiplied by $\alpha^4$. According to Equation~(\ref{eq2.3}), we have the following formulae for some integer $l$:
	\begin{equation*}
		\begin{aligned}
			x_{nP}&=x_P-\frac{\alpha^{2l}W(n-1)W(n+1)}{\alpha^{2l}W(n)^2},\\
			&=x_P-\frac{W(n-1)W(n+1)}{W(n)^2},\\		
			y_{nP}&=\frac{\alpha^{3l}(W(n-1)^2W(n+2)-W(n+1)^2W(n-2))}{\alpha^{3l}4y_PW(n)^3},\\
			&=\frac{W(n-1)^2W(n+2)-W(n+1)^2W(n-2)}{4y_PW(n)^3}.
		\end{aligned}
	\end{equation*}\qed
\end{proof}
Besides, we can have a further improvement based on the algorithm in \cite{8890309}. Their Double-and-Add algorithm is improved by using some tricks to save $2$ multiplications at each iteration, but it involves $6$ right-shift operations. We can replace these operations by $2$ left-shift operations.

\section{The Elliptic Net Algorithm on the Twisted Curve}\label{sec4}
The application of the twisted curve brings some significant improvements in Miller's algorithm. However, if we use the twist trick on the original curve with the Elliptic Net algorithm, then it will not be as portable and intuitive as Miller's algorithm.
In 2010, Costello \emph{et al.} \cite{2010Faster} proposed the Ate pairing entirely on the twisted curve for Miller's algorithm. Actually, when we use the Elliptic Net algorithm to compute pairings, it will also have a good improvement if the related parameters all on the twisted curve. In this section, we will analyze the Ate pairing and the Optimal Ate pairing on the twisted curve with our method and apply our work to the Elliptic Net algorithm.
\subsection{The Ate Pairing on the Twisted Curve}
Let $E$ be an elliptic curve over $\mathbb{F}_q$, and the related parameters are defined in Section~\ref{sec2.1}. Let $E'/\mathbb{F}_{q^{e}}$ be the twist of $E$ of degree $d$ with $e=k/d$. Let $\pi'_{q^e}$ be the $q^e$-power Frobenius map on $E'$. There exists an isomorphism $\Psi_d:E'\rightarrow E$ over $\mathbb{F}_{q^k}$, then we can define two groups
$$\mathbb G_1' \triangleq E'[r]\bigcap Ker(\pi' _{q^e} -[1]),\,\mathbb G_2' \triangleq E'[r]\bigcap Ker(\pi' _{q^e} -[q^{e}]).$$
Actually, the iterations of the Miller loop $T$ can be set as $(t-1)\,mod\,r$~\cite{Matsuda2007Optimised}. When we compute the Ate pairing, the operations are all on the original curve. In the following part, we will give a new derivation of the theorem about pairings entirely on the twisted curve.
\begin{theorem}\label{th3.1}
	For $\Psi_d^{-1}(P) \in \mathbb{G}_2',\,Q' \in \mathbb{G}_1'$, we can define a pairing on $\mathbb{G}_1' \times \mathbb{G}_2'$ if $r\nmid (T^k-1)/r$:
	$$\begin{aligned}
		Ate_{E'}:\mathbb G_1' \times \mathbb G_2' &\rightarrow \mu_r\\ (Q',\Psi_d^{-1}(P))&\mapsto Ate_{E'}(Q',\Psi_d^{-1}(P))=(f_{T,Q'}(\Psi_d^{-1}(P)))^{{q^k-1}/r}.
	\end{aligned}$$
\end{theorem}
\begin{proof}
	We only need to prove that $f_{T,\Psi_d(Q')}=f_{T,Q'}\circ \Psi_d^{-1}$, for all $Q'\in \mathbb{G}_1'$.
	
	The divisor of $f_{T,\Psi_d(Q')}$ is
	$${\rm Div}(f_{T,\Psi_d(Q')})=T(\Psi_d(Q'))-([T]\Psi_d(Q'))-(T-1)(\infty),$$
	and since $\Psi_d$ is an isomorphism,
	$$\begin{aligned}
		(\Psi_d)^*{\rm Div}(f_{T,\Psi_d(Q')})&=T(\Psi_d)^*(\Psi_d(Q'))-(\Psi_d)^*([T]\Psi_d(Q'))-(T-1)(\Psi_d)^*(\infty),\\
		&=T(Q')-([T]Q')-(T-1)(\infty),\\
		&=(f_{T,Q'}).
	\end{aligned}$$
	Furthermore, we have $(\Psi_d)^*{\rm Div}(f_{T,\Psi_d(Q')})={\rm Div}(f_{T,\Psi_d(Q')}\circ \Psi_d)$, then we have $f_{T,\Psi_d(Q')} \circ \Psi_d=f_{T,Q'}$. We compose the formula with $\Psi_d^{-1}$ on both sides, and obtain:
	$$f_{T,\Psi_d(Q')}(P)=f_{T,Q'}\circ \Psi_d^{-1}(P).$$\qed
\end{proof}

\begin{remark}
	When we compute pairings on the twisted curves, the operations are always in the field where $Q'$ is located. The final value we need can be obtained by twists. The transformation involved here is very small for Miller's algorithm. This is because each transformation only needs to be multiplied by a fixed value $\alpha$ on $\mathbb{F}_{q^k}$. Generally, $\alpha$ is sparse. But for the Elliptic Net algorithm, if we adopt the same idea to use this isomorphism, then the value of $\alpha$ will be changed as the iterations, which means that the transformation of the value we need will not be a friendly process. Therefore, we choose to compute pairings on the twisted curve for the Elliptic Net algorithm.
\end{remark}

\subsection{The Optimal Ate Pairing on the Twisted Curve}
For the Optimal Ate pairing on the twisted curve, the situation is more complicated than that of the Ate pairing while we can still derive the following theorem easily.
\begin{theorem}\label{the3.3}
	Let $\lambda = mr$ with $r\nmid m$ and $\lambda = \sum_{i=0}^{\varphi(k)}c_iq^i$. Define  $$\Phi_{d,i}=\Psi_d^{-1}\circ [c_iq^i]\circ \Psi_d,$$ and note that $\Phi_{d,s_i}=\Psi_d^{-1}\circ [s_i]\circ \Psi_d$, where $s_i=\sum_{j=i}^{\varphi(k)}c_jq^j$.  There exists a pairing on $\mathbb{G'}_1 \times \mathbb{G'}_2$:
	$$\begin{aligned}
		Opt_{E'}:\mathbb{G'}_1 \times \mathbb{G'}_2 &\rightarrow \mu_r\\
		(Q',\Psi_d^{-1}(P))&\!\longmapsto (\!\prod_{i=0}^{\varphi(k)}\!f_{c_i,Q'}^{q^i}(\Psi_d^{-1}\!(P))\!\cdot \!\prod_{i=0}^{\varphi(k)-1}\!\frac{l_{\Phi_{d,s_{i+1}} ,\Phi_{d,i}(Q')}}{v_{\Phi_{d,s_i} (Q')}}(\Psi_d^{-1}(P)))^{(q^k-1)/r}.
	\end{aligned}
	$$
\end{theorem}
\begin{proof}
	From Theorem~\ref{th3.1}, we have
	$${\rm Div}(\prod_{i=0}^{\varphi(k)}f_{c_i,Q'}^{q^i}\circ \Psi_d^{-1})={\rm Div}(\prod_{i=0}^{\varphi(k)}f_{c_i,\Psi_d(Q')}).$$
	Let $Q_i\triangleq [s_{i+1}]\circ \Psi_d(Q')$. Consider the relation between  $l_{\Phi_{d,s_{i+1}} ,\Phi_{d,i}(Q')}$ and $l_{Q_i,[c_iq^i]\Psi_d(Q')}$. 
	From the definition of divisors,
	$$
	{\rm Div}(l_{Q_i,[c_iq^i]\Psi_d(Q')})=(Q_i)+([c_iq^i]\Psi_d(Q'))+(-Q_{i+1})-3(\infty).$$
	Since $\Psi_d$ is an isomorphism,
	$$\begin{aligned}
		\Psi_d^*{\rm Div}(l_{Q_i,[c_iq^i]\Psi_d(Q')})&=(\Psi_d^{-1}(Q_i))+(\Psi_d^{-1}\circ [c_iq^i]\circ \Psi_d(Q'))+(-Q_{i+1})-3(\infty),\\
		&=(l_{\Phi_{d,s_{i+1}} ,\Phi_{d,i}(Q')}).\\
	\end{aligned}$$
	Therefore,
	$$l_{Q_i,[c_iq^i]\Psi_d(Q')}(P)=l_{\Phi_{d,s_{i+1}} ,\Phi_{d,i}(Q')}\circ \Psi_d^{-1}(P).$$
	Similarly,
	$$v_{Q_i}(P)=v_{\Phi_{d,i}}(Q')\circ \Psi_d^{-1}(P).$$\qed
\end{proof}

\begin{remark}
	For $Q'\in \mathbb{G}_1'$, we have $\pi_q\circ\Psi_d(Q')=[q]\Psi_{d}(Q')$.
	
	Since $\pi_q$ is an endomorphism and $\Psi_d$ is an isomorphism over $\mathbb{F}_{q^k}$, we have
	$$\pi_q\circ\Psi_d(Q')=\Psi_{d}\circ[q](Q').$$
	Therefore, we have 
	$$\Psi_d^{-1}\circ\pi_q\circ\Psi_{d}(Q')=[q](Q'),i.e., \Phi_{d,1}(Q')=[q](Q').$$
\end{remark}

Thus, we know that the point in $\mathbb{G}_1'$ can be mapped to a point in $E[r]$. This also means that the line function on the twisted curve via this result. 

Note that the ratio of the cost of inversions to the multiplications over $\mathbb{F}_{q^k}$ decreases if the size of $\mathbb{F}_{q^k}$ is larger. When we compute the Ate pairing on the twisted curve, our operations in the first dimension vector centered on $i$ are in $\mathbb{F}_{q^e}$. Compared with the operations in $\mathbb{F}_{q^k}$, it is more necessary to eliminate the inverse operation when the bit is not equal to $0$. Furthermore, we can use the NAF form to ensure that the density $\rho$ is within the effective range to accelerate the IENA.

\section{The Elliptic Net Algorithm with Lazy Reduction}\label{sec5}
Lazy reduction technique can also be applied to speed up the Elliptic Net algorithm. Lazy reduction was presented formally in \cite{lim2000fast}. It can save the number of modular reductions during the calculation. The main idea of lazy reduction is to put the required modular reductions of some multiplication operations like $\sum a_ib_i$ over $\mathbb{F}_q$ to the end. So these multiplication operations only need $1$ modular reduction over $\mathbb{F}_q$. Thus it can save the number of modular reductions during the calculation. In this paper, we use Montgomery reduction \cite{Montgomery1985Modular}, so the cost of a modular reduction is equal to the cost of one multiplication. Note that each item of $a_ib_i$ without modular reduction should satisfy the upper bound of Montgomery reduction.

When we use the Elliptic Net algorithm to compute pairings, it contains lots of multiplications like $A \cdot B \pm C \cdot D$, which needs 2 modular reductions normally. But if we use lazy reduction, we can only need one modular reduction. Obviously, we are not concern about violating the upper bound for this situation since we only use the lazy reduction once each time, and we set $A,B,C,D\in \mathbb{F}_q$. The proposed algorithms using lazy reduction are given for the initialization step and Double-and-Add step respectively.
We mainly improve the term $W(3,0)$ and $W(4,0)$ at the initialization step. The improvement is not obvious here, so we only give the number of modular reductions of three situations in Table~\ref{tab1}.
\begin{table}
	\caption{The Number of Modular Reductions at the Initialization Step}
	\label{tab1}
	\begin{center}
		\begin{tabular}{|l|l|l|l|}
			\hline
			Algorithm &  $A,B \neq 0$ & $B = 0$ & $A = 0$\\
			\hline
			ENA \cite{Stange2007} &  10 & 8 & 6\\
			This work & 7 & 6 & 5\\
			\hline
		\end{tabular}
	\end{center}
\end{table}

The explicit updating formulas at the Double-and-Add step are mentioned in Section~\ref{sec2.3}. The $Double(V)$ and $DoubleAdd(V)$ functions are combined with the lazy reduction technique, and we adopt 
the new Double-and-Add step in \cite{Chen2015AnIO} which needs 10 terms in total. We present the Double-and-Add algorithm based on the IENA in Appendix~\ref{ap1}. Assume that our terms belong to the finite field $\mathbb F_q$. At step [7]-[23] we compute the $Double(V)$ function. We update 7 terms in the first vector and 3 terms in the second vector that are both centered on $2i$. In the ENA, we need 42 modular reductions in each iteration. In contrast, the number of modular reductions decreases to 37 in the IENA. With the help of lazy reduction, the updating process of each term can save one modular reduction, so 10 terms will save 10 modular reductions in total. The $DoubleAdd(V)$ function is computed at step [25]-[43]. These steps contain 40 modular reductions originally and the number of modular reductions decreases to 30 with lazy reduction in each iteration. Table~\ref{tab2} shows the number of modular reductions of three Elliptic Net Algorithms at the Double-and-Add step. 
\begin{table}
	\caption{The Number of Modular Reductions at the Double-and-Add Step}
	\label{tab2}
	\begin{center}
		\begin{tabular}{|l|l|l|}
			\hline
			Algorithm &  $Double(V)$ & $DoubleAdd(V)$\\
			\hline
			ENA \cite{Stange2007} &  42 & 42\\
			IENA \cite{Chen2015AnIO} & 37 & 40\\
			This work & 27 & 30\\
			\hline
		\end{tabular}
	\end{center}
\end{table}

\section{Implementation and Analysis}\label{sec6}
In this section, we implement the optimization of the Elliptic Net algorithm for scalar multiplication and pairing computation respectively. Our algorithms are implemented by using the C programming language compiled with the GCC compiler of which the version is 7.4.0. Our code is based on version 0.5.0 of the RELIC toolkit \cite{relic-toolkit} and we use the Intel Core i7-8550U CPU processor operating at 1.80 GHz that runs on a 64-bit Linux. 
Our implementation will be divided into two parts. One is scalar multiplication, and the other one is pairing computation. We apply lazy reduction technique to both parts. Note that lazy reduction has a good acceleration effect in the Elliptic Net algorithm.

In the firsy part, we will use different methods to compare the efficiency of computing the Optimal Ate pairing on the twisted curve at 128-bit security level and 192-bit security level respectively. Notice that the $DoubleAdd(V)$ function is not friendly in the IENA. In general, we will choose the loop length which has a low Hamming weight so that we can use $Double(V)$ function more frequently in the whole iterations to accelerate the algorithm. 

In the second part, we will show the comparison of the efficiency between computing scalar multiplication in \cite{8890309} and our work. Scalar multiplication algorithm with division polynomials is similar to the ladder algorithm, and this algorithm is easily coded compared to the traditional double-and-add algorithm. It can naturally resist power attacks, but it is slower than the basic double-and-add algorithm. Therefore, we do not compare our algorithms with the state-of-the-art algorithm for standard elliptic curve scalar multiplication algorithm~\cite{fastecc}. We choose the NIST P-384 curve and the NIST P-521 curve to compute scalar multiplication respectively \cite{Hankerson2011}. Notice that for the NIST P-384 curve, the prime $p$ satisfies $\gcd(p-1,3)=1$. Therefore, we can combine works in \cite{2017A} with our work to get a further improvement on this curve.

The elliptic curves we choose are the 381-bit BLS12 and 676-bit KSS18 curves. We specify some symbols here to show the amount of operations in this section:
\begin{itemize}
	\item $M_{k}$: the multiplication over $\mathbb{F}_{q^{k}}$, $S_{k}$: the square operation over $\mathbb{F}_{q^{k}}$,
	\item $M$: the multiplication over $\mathbb{F}_{q}$,
	$S$: the square operation over $\mathbb{F}_{q}$,
	\item $I_k$: the inversion over $\mathbb{F}_{q^{k}}$,
	$A$: the addition operation over $\mathbb{F}_{q}$.
\end{itemize}

\subsection{Pairing Computation}
In the following part, we will focus on the improvement of pairing computation using the Elliptic Net algorithm.
\subsubsection{381-bit BLS12 Curve }\label{sec6.1}

The concrete parameters for the 381-bit BLS12 curve with embedding degree $k=12$ are given as follows.

\begin{itemize}
	\item $t=-2^{63} - 2^{62} - 2^{60} - 2^{57} - 2^{48} - 2^{16}$;
	\item $r = t^4- t^2 + 1$;
	\item $q = (t-1)^2(t^4-t^2+1)/3+t$;
	\item $E_0:y^2=x^3+4$ over $\mathbb F_q$;
	\item $\mathbb F_{q^2}=\mathbb F_q[u]/\langle u^2-\beta \rangle$, where $\beta = -1$;
	\item $\mathbb F_{q^6}=\mathbb F_{q^2}[v]/\langle v^3-\xi \rangle$, where $\xi = u+1$;
	\item $\mathbb F_{q^{12}}=\mathbb F_{q^6}[\omega]/\langle \omega^2-v \rangle $;
	\item the twisted curve $E_0'$ : $y^2=x^3+4\xi$ over $\mathbb F_{q^2}$.
\end{itemize}

Recall that $P\in E_0(\mathbb{F}_q)$ and $Q'\in E_0'(\mathbb{F}_{q^{e}})$. We apply three techniques discussed in this work to the ENA and the IENA for computing the Optimal Ate pairing. According to Theorem~\ref{the3.3} in Section~\ref{sec3}, the explicit formulas of line functions on the twisted BLS12 curve can be obtained. Therefore, for the BLS12 curve, we only need to calculate $$(f_{t,\Psi_6(Q')}(P))^{\frac{q^{12}-1}{r}}\,\, or \,\, (f_{t,Q'}(\Psi_6^{-1}(P))) ^{\frac{q^{12}-1}{r}}.$$
The amount of operations for $f_{t,\Psi_6(Q')}$ and $f_{t,Q'}$ in one iteration is $7S_{12}+\frac{67}{2}M_{12}$ and $6S_2+62M_2+S_{12}+\frac{3}{2}M_{12}$ at the Double step in the ENA, respectively. Note that in our implementation, $1I_{12}\approx 3M_{12}$, $1I_2\approx 13M_2$, $1M_2\approx3M$ and $1M_{12}\approx 54M$. In the IENA, we need $6S_{12}+31M_{12}+I_{12}$ without twist when the bit is not equal to $0$. If we compute pairings on its corresponding twisted curve, the operations can be reduced to $1S_{12}+1M_{12}+5S_2+39M_2+1I_2$. Without considering the influence of delay error, it is not necessary to eliminate the inverse operation if we do not use twists of elliptic curves here. But when the cost of one inversion is greater than the cost of 5 multiplications, eliminating the inverse operation can have a more obvious improvement. Moreover, since $t$ is a negative number, we choose to compute $f_{-t,Q'}$ and use the relationship $(f_{t,Q'})^{(q^{12}-1)/r}=(\dfrac{1}{f_{-t,Q'}})^{(q^{12}-1)/r}$ to revise the value. Note that in order to make the IENA work well, we choose to expand $-t$ in the NAF form to reduce the proportion of non-zero digits. Then the non-zero digits density $\rho$ will be smaller than that of the previous one. Although the Elliptic Net algorithm is much slower than Miller's algorithm, it still counts in milliseconds. Therefore, we cycle the program $10,000$ times and take the average value to ensure the stability and accuracy of our program.
The comparison about the efficiency of different methods is provided in Table~\ref{tab4}. 

\begin{table}
	\caption{Efficiency Comparison on a 381-bit BLS12 Curve}
	\label{tab4}
	\begin{center}
		\begin{tabular}{|l|l|l|}
			\hline
			Method &  Clock Cycle ($\times 10^3$)& Time (ms)\\
			\hline
			ENA \cite{Stange2007} &  {25,524} & {12.81}\\
			ENA  with lazy reduction &  {24,599} & {12.35}\\
			IENA \cite{Chen2015AnIO} &  {23,508} & {11.80}\\
			IENA with lazy reduction &  {22,586} & {11.34}\\
			IENA (Eliminate Inverse) &  {23,554} & {11.82}\\
			IENA (Eliminate Inverse) with lazy reduction&  {22,722} & {11.41}\\
			ENA  (Twist) &  {4,890} & {2.45}\\
			ENA  (Twist) with lazy reduction&  {4,463} & {2.24}\\
			IENA(Twist) &  {4,749} & {2.38}\\
			IENA(Twist) with lazy reduction&  {4,325} & {2.17}\\
			IENA(Twist $\&$ Eliminate Inv) &  {4,575} & {2.30}\\
			IENA(Twist $\&$ Eliminate Inv) with lazy reduction&  {4,315} & {2.16}\\
			Miller's algorithm & 3,123 & 1.57 \\
			\hline
		\end{tabular}
	\end{center}
\end{table}

From Table~\ref{tab4}, we can see that this work speeds up the Elliptic Net algorithm indeed and the efficiency of computing the Optimal Ate pairing on the twisted curve is much quicker than that on the original elliptic curve. The twist technology has a good performance for both algorithms. The efficiency has been increased by about $80.9\%$ without using lazy reduction in the IENA. Notice that lazy reduction also plays a vital role in the algorithm, which further accelerates the algorithm. Besides, the elimination of the inversion has also been proved to be effective which is up to $3.36\%$ faster than the IENA. 
Compared to the ENA, the efficiency of our work on the original and twisted curves increases by around  $11\%$ and $11.8\%$, respectively.
\subsubsection{676-bit KSS18 Curve}\label{sec6.3}

Now we give the parameters of the 676-bit KSS18 curve with embedding degree $k=18$ below:

\begin{itemize}
	\item $t=-2^{85}-2^{31}-2^{26}+2^6$;
	\item $r = (t^6+37t^3+343)/343$;
	\item $q = (t^8+5t^7+7t^6+37t^5+188t^4+259t^3+343t^2+1763t+2401)/21$;
	\item $E_0:y^2=x^3+2$ over $\mathbb F_q$;
	\item $\mathbb F_{q^3}=\mathbb F_q[u]/\langle u^3-\beta \rangle$, where $\beta = -2$;
	\item $\mathbb F_{q^6}=\mathbb F_{q^2}[v]/\langle v^2-\xi \rangle$, where $\xi = u$;
	\item $\mathbb F_{q^{18}}=\mathbb F_{q^6}[\omega]/\langle \omega^3-v \rangle $;
	\item the twisted curve $E_0'$ : $y^2=x^3+2/\xi$ over $\mathbb F_{q^2}$.
\end{itemize}

We need to calculate $$(f_{t,\Psi_6(Q')}\cdot f_{3,\Psi_6(Q')}^q\cdot l_{[t]\Psi_6(Q'),[3q]\Psi_6(Q')}(P))^{\frac{q^{18}-1}{r}}$$ or $$(f_{t,Q'}\cdot f_{3,Q'}^q\cdot l_{\Psi_6^{-1}\circ [t]\circ \Psi_6(Q'),\Psi_6^{-1}\circ [3q]\circ \Psi_6(Q')}(\Psi_6^{-1}(P)))^{\frac{q^{18}-1}{r}}$$  for computing the Optimal Ate pairing on this curve. In order to make our comparisons more obviously and steadily, we calculate the Optimal Ate pairing $1,000$ times, and take the average value as the final result. Table~\ref{tab5} shows the timings of different methods for computing the Optimal Ate pairing.
\begin{table}
	\caption{Efficiency Comparison on a 676-bit KSS Curve}
	\label{tab5}
	\begin{center}
		\begin{tabular}{|l|l|l|}
			\hline
			Method &  Clock Cycle ($\times 10^3$)& Time (ms)\\
			\hline
			ENA  \cite{Stange2007} &  {136,542} & {68.54}\\
			ENA  with lazy reduction &  {132,700} & {66.61}\\
			IENA \cite{Chen2015AnIO} &  {122,629} & {61.56}\\
			IENA with lazy reduction &  {119,991} & {60.23}\\
			IENA (Eliminate Inverse) &  {122,681} & {61.59}\\
			IENA (Eliminate Inverse) with lazy reduction&  {120,686} & {60.58}\\
			ENA  (Twist) &  {40,949} & {20.56}\\
			ENA  (Twist) with lazy reduction&  {39,440} & {19.80}\\
			IENA(Twist) &  {40,676} & {20.42}\\
			IENA(Twist) with lazy reduction&  {39,276} & {19.72}\\
			IENA(Twist $\&$ Eliminate Inv) &  {40,291} & {20.23}\\
			IENA(Twist $\&$ Eliminate Inv) with lazy reduction&  {38,904} & {19.53}\\
			Miller's algorithm & {17,149} & {8.61} \\
			\hline
		\end{tabular}
	\end{center}
\end{table}

On the KSS18 curve, the effect of our modification is similar to the performance on the BLS12 curve. Just comparing the performance of the ENA on the twisted curve and the original curve, the algorithm is $70\%$ faster on the twisted curve. But after eliminating the inverse operation and using lazy reduction technique, the algorithm can be about $5\%$ faster than the IENA on the twisted curve.

From these results, we find that the improvement of lazy reduction on the KSS18 curve is increased. This is mainly because the embedding degree on the KSS18 curve is bigger than that of the BLS12 curve. Besides, we have more iterations of the Miller loop on the KSS18 curve. But the amount of optimization in a single iteration is same. In contrast to our theory, the efficiency of computing the Optimal Ate pairing on the twisted curve is much higher than that on the original curve for the Elliptic Net algorithm. In addition, we can further improve the efficiency of the algorithm by eliminating the inverse operation. Notice that Miller's algorithm performs well in our implementation with the cost time of $1.57\,ms$ on the 381-bit BLS12 curve. Its version in our work is the fastest one implemented by Diego \emph{et al.} in the Relic toolkit~\cite{relic-toolkit}, and we test its efficiency in our personal computer. However, compared with the previous work, the gap between the Elliptic Net algorithm and Miller's algorithm has been greatly shortened, which from the original cost of more than 9 times to the current cost of less than 2 times.

\subsection{Scalar Multiplication}
Our work based on the scalar multiplication algorithm proposed in \cite{8890309} is to replace $6$ right-shift operations by $2$ left-shift operations in each iteration. It seems that this improvement will not work obviously. However, after we use the lazy reduction technique, the efficiency will have a good improvement. We choose two curves which achieve 192-bit security level and 256-bit security level respectively. The equations of these curves over $\mathbb{F}_q$ have the form: $y^2=x^3-3x+b$.
\subsubsection{NIST P-521 Curve}
In \cite{8890309}, the amount of operations is $24M+6S+36A$ and $6$ right-shift operations at each iteration. Let one subtraction operation or one double operation be equal to one addition operation. The trick in Section~\ref{sec3} is applied to this algorithm, and then we replace $6$ right-shift operations by $2$ left-shift operations. Table~\ref{tab8} provides the timings of scalar multiplication algorithm in \cite{8890309} and this work.

\begin{table}
	\caption{Efficiency of Scalar Multiplication on the NIST P-521 Curve}
	\label{tab8}
	\begin{center}
		\begin{tabular}{|l|l|l|}
			\hline
			Method &  Clock Cycle ($\times 10^3$) & Time (ms)\\
			\hline
			Algorithm \cite{8890309} &  {5,097} & {2.56}\\
			Algorithm \cite{8890309} with lazy reduction &  {4,844} & {2.43}\\
			This work &  {4,920} & {2.47}\\
			This work with lazy reduction &  {4,530} & {2.27}\\
			\hline
		\end{tabular}
	\end{center}
\end{table}
Our work facilitates an acceleration of around $11.2\%$ over the algorithm in \cite{8890309} of scalar multiplication.
However, the efficiency of these algorithms is slower than that of the ENA for scalar multiplication except the prime $p$ is large enough.
\subsubsection{NIST P-384 Curve}
We focus on the situation of $gcd(p-1,3)=1$ and combine the work in \cite{2017A} and \cite{8890309} to compute scalar multiplication. Let $\alpha\in \mathbb{F}_q $ such that $\alpha^3=W(2)^{-1}$. Then the initial values of an elliptic net are given below:
\begin{equation*}
	\begin{aligned}
		&\tilde{W}(1)=1,\tilde{W}(2)=1,\tilde{W}(3)=\alpha^8 \cdot W(3),\\
		&\tilde{W}(4)=\alpha^{15}\cdot W(4),\tilde{W}(5)=\tilde{W}(4)-\tilde{W}(3)^3.
	\end{aligned}
\end{equation*}
We use these new initial values above to compute scalar multiplication. The amount of operations will be reduced from $20M+6S+36A$ and $6$ right-shift operations to $18M+6S+36A$ and $2$ left-shift operations in each iteration. Table~\ref{tab9} reflects the efficiency of algorithm in \cite{8890309} and our work for computing scalar multiplication on the NIST P-384 curve. Results shown that we have an improvement based on \cite{8890309} with $14.96\%$.
\begin{table}
	\caption{Efficiency of Scalar Multiplication on the NIST P-384 Curve}
	\label{tab9}
	\begin{center}
		\begin{tabular}{|l|l|l|}
			\hline
			Method &  Clock Cycle ($\times 10^3$) & Time (ms)\\
			\hline
			Algorithm \cite{8890309}  &  {2,329} & {1.17}\\
			Algorithm \cite{8890309} with lazy reduction &  {2,208} & {1.11}\\
			This work &  {2,234} & {1.12}\\
			This work with lazy reduction &  {1,980} & {0.99}\\
			\hline
		\end{tabular}
	\end{center}
\end{table}
\section{Conclusions}\label{sec7}
In this work, we improved the Elliptic Net algorithm. Among different versions of the Elliptic Net algorithm, we analyzed their efficiency and presented higher speed records on the computation of the Optimal Ate pairing on a 381-bit BLS12 curve and a 676-bit KSS18 curve by using the Elliptic Net algorithm with several tricks, respectively. We also improved the scalar multiplication algorithm in \cite{8890309} and implemented our work on a NIST P-384 curve and a NIST P-521 curve, respectively. The scalar multiplication algorithm was increased by up to $14.96\%$ than the work in \cite{8890309}. The lazy reduction technique was able to reduce by around $27\%$ of the required modular reductions. Moreover, the application of twist technology helped us reduce the number of multiplications and the improvement was significant. Besides, the improved Elliptic Net algorithm was also further improved, i.e., the inverse operation can be replaced by few multiplications when the bit is equal to $1$ or $-1$.  On the 381-bit BLS12 curve, this work improved the performance of the Optimal Ate pairing by $80\%$ compared with the original version on a 64-bit Linux platform. The implementation on the 676-bit KSS18 curve had shown that this work was $71.5\%$ faster than the previous ones. 
Our results shown that the Elliptic Net algorithm can compute pairings efficiently on personal computers while it was still slower than Miller's algorithm. In the future, we will consider the parallelization of the Elliptic Net algorithm to get a further improvement.


%% file: tex/appendix.tex
\clearpage
\appendix
\addappheadtotoc

	\section{Algorithm in This Work}\label{ap1}
	\floatname{algorithm}{Algorithm}
	\renewcommand{\algorithmicrequire}{\textbf{INPUT:}}
	\renewcommand{\algorithmicensure}{\textbf{OUTPUT:}}
	\begin{breakablealgorithm}
		\caption{Double-and-Add Algorithm with Lazy Reduction (Eliminate Inversion)}
		\label{alg5}
		\begin{algorithmic}[1] 
			\REQUIRE Block $V$ centered on $i$ in which the first vector has 7 terms and the second vector has 3 terms. 
			$W(1,0)=W(0,1)=1$, $\alpha=W(2,0)^{-1}$, $\beta=W(-1,1)^{-1}$, $\gamma_1=W(2,-1)^{-1}$, $\delta=W(1,1)^{-1}$, $\omega_2=W(2,0)^2$, $\omega_{13}=W(1,0)W(3,0)$, $flag\in \left\{0,1\right\}.$
			\ENSURE Block centered on $2i$ if $flag=0$, centered on $2i+1$ if $flag=1$.
			\STATE $S_0 \gets V[2,2]^2\,mod\,p$, $P_0 \gets (V[2,1]*V[2,3])\,mod\,p$; //$1S_k+1M_k$
			\FOR{$i = 1 \to 5$}
			\STATE $S[i] \gets V[1,i+1]^2 \,mod\,p$, $P[i] \gets (V[1,i]*V[1,i+2])\,mod\,p$;
			\ENDFOR
			\IF {$flag=0$}
			\FOR{$j = 1 \to 3$}
			\STATE $t_0 \gets S[j]*P[j+1]$, $t_1 \gets S[j+1]*P[j]$, $V[1,2j-1] \gets (t_0-t_1)\,mod\,p$;
			\STATE $t_0 \gets S[j]*P[j+2]$, $t_1 \gets S[j+2]*P[j]$, $V[1,2j] \gets (t_0-t_1)\,mod\,p$;
			\STATE $V[1,2j] \gets (V[1,2j]*\alpha)\,mod\,p$;
			\ENDFOR
			\STATE $t_0 \gets S[4]*P[5]$, $t_1 \gets S[5]*P[4]$, $V[1,7] \gets (t_0-t_1)\,mod\,p$;
			\STATE $k_0 \gets S[2]*P_0$, $k_1 \gets P[2]*S_0$, $V[2,1] \gets (k_0-k_1)\,mod\,p$;
			\STATE $V[2,1] \gets (V[2,1]*\delta)\,mod\,p$;
			\STATE $k_0 \gets S[3]*P_0$, $k_1 \gets P[3]*S_0$, $V[2,2] \gets (k_0-k_1)\,mod\,p$;
			\STATE $k_0 \gets S[4]*P_0$, $k_1 \gets P[4]*S_0$, $V[2,3] \gets (k_0-k_1)\,mod\,p$;
			\STATE $V[2,3] \gets (V[2,3]*\beta)\,mod\,p$;
			\ELSE
			\FOR{$j = 1 \to 3$}
			\STATE $t_0 \gets S[j]*P[j+2]$, $t_1 \gets S[j+2]*P[j]$, $V[1,2j-1] \gets (t_0-t_1)\,mod\,p$;
			\STATE $V[1,2j-1] \gets (V[1,2j-1]*\alpha)\,mod\,p$;
			\STATE $t_0 \gets S[j+1]*P[j+2]$, $t_1 \gets S[j+2]*P[j+1]$, $V[1,2j] \gets (t_0-t_1)\,mod\,p$;
			\ENDFOR
			\STATE $vt_1 \gets (V[1,4]*V[1,6])\,mod\,p$, $vt_2 \gets (V[1,5]^2)\,mod\,p$; //$1M_e+1S_e$
			\STATE $t_0 \gets vt_1*\omega_2$, $t_1 \gets vt_2*\omega_{13}$, $V[1,7] \gets (t_0-t_1)\,mod\,p$; //$2M_e$
			\FOR{$j = 1 \to 6$}
			\STATE $V[1,j]=(V[1,j]*V[1,3])\,mod\,p$;
			\ENDFOR
			\STATE $k_0 \gets S[3]*P_0$, $k_1 \gets P[3]*S_0$, $V[2,1] \gets (k_0-k_1)\,mod\,p$;    
			\STATE $k_0 \gets S[4]*P_0$, $k_1 \gets P[4]*S_0$, $V[2,2] \gets (k_0-k_1)\,mod\,p$;
			\STATE $V[2,2] \gets( V[2,2]*\beta)\,mod\,p$;
			\STATE $k_0 \gets S[5]*P_0$, $k_1 \gets P[5]*S_0$, $V[2,3] \gets (k_0-k_1)\,mod\,p$;
			\STATE $V[2,3] \gets (V[2,3]*\gamma_1)\,mod\,p$.
			\ENDIF
			\RETURN{$V$}
		\end{algorithmic}
	\end{breakablealgorithm}